\documentclass[12pt]{article}

\usepackage{amsmath,amssymb,amsfonts,amsthm}
\usepackage{fullpage}
\usepackage{float}

\usepackage{lipsum}
\usepackage{algorithmicx}
\usepackage{algpseudocode}

\theoremstyle{plain}
\newtheorem{theorem}{Theorem}
\newtheorem{corollary}[theorem]{Corollary}
\newtheorem{lemma}[theorem]{Lemma}
\newtheorem{proposition}[theorem]{Proposition}

\theoremstyle{definition}

\newtheorem{openproblem}[theorem]{Open Problem}

\theoremstyle{remark}

\title{Remarks on Privileged Words}

\author{Michael Forsyth, Amlesh Jayakumar, and Jeffrey Shallit\\
School of Computer Science \\
University of Waterloo \\
Waterloo, ON  N2L 3G1 \\
Canada \\
{\tt \{mfforsyth,a3jayakumar,shallit\}@uwaterloo.ca}
}

\begin{document}

\maketitle

\begin{abstract}
We discuss the notion of privileged word, recently introduced
by Peltom\"aki.  
A word $w$ is privileged if it is of length $\leq 1$,
or has a privileged border that occurs exactly twice in $w$.
We prove the following results:
(1) if $w^k$ is privileged for some $k \geq 1$, then $w^j$ is
privileged for all $j \geq 0$;
(2) the language of privileged words is
neither regular nor context-free;
(3) there is a linear-time algorithm to check if a given word is
privileged; and 
(4) there are at least $2^{n-5}/n^2$
privileged binary words of length $n$.
\end{abstract}

\section{Introduction}

We say that a word $x$ is a {\it border} of $w$ if it is
both a prefix and a suffix of $w$. 

Peltom\"aki \cite{Pelt1,Pelt2} recently introduced the notion
of {\it privileged word}.
A word $w$ is privileged if
\begin{itemize}
\item[(a)] it is of length $\leq 1$, or
\item[(b)] it has a privileged
border that appears exactly twice in $w$.
\end{itemize}

Here are the first few privileged words over a binary alphabet:
$$  
0, 1, 00, 11, 000, 010, 101, 111,
0000,0110,1001,1111,
00000,00100,01010,01110,10001, $$
$$ 10101,11011,11111,
000000,001100,010010,011110,100001,101101,110011,111111. $$
An easy induction shows that $a^i$ is privileged for
for any letter $a$ and $i \geq 0$.    

We now recall two results of Peltom\"aki \cite{Pelt1}.

\begin{theorem}
Let $w$ be privileged.
\begin{itemize}
\item[(a)]  If $t$ is a privileged prefix (resp., suffix) of $w$, then
$t$ is also a suffix (resp., prefix) of $w$.

\item[(b)]  If $v$ is a border of $w$ then $v$ is privileged.
\end{itemize}
\label{pelto}
\end{theorem}

Define the {\it number of leading $a$'s in $w$} to be the
largest integer $n$ such that $a^n$ is a prefix of $w$, and
similarly for the number of trailing $a$'s.  Then we have

\begin{corollary}
If $w$ is privileged, then the number of leading $a$'s in $w$
equals the number of trailing $a$'s.
\label{cor-leading}
\end{corollary}

\begin{proof}
Write $w = a^i z a^j$ where $z$ neither begins nor ends in $a$.
Then by Theorem~\ref{pelto} (a) we see that $i \geq j$ and $j \geq i$.
\end{proof}

We now state a useful lemma.

\begin{lemma}
Let $w$ be a nonempty word.  Then $w$ is privileged if and only
if its longest proper privileged prefix is also a suffix of $w$.
\label{fors}
\end{lemma}

\begin{proof}
$\Longrightarrow$:  follows from Theorem~\ref{pelto} (a) above.  

$\Longleftarrow$: Let $u$ be the longest proper privileged
prefix of $w$.  Let $v$ be the shortest prefix of $w$ containing
exactly two occurrences of $u$; this is well-defined 
since $u$ is a suffix of $w$.  Then $v$ itself is privileged.
So either $v = w$, or $|u| < |v| < |w|$ and $v$ is a longer
proper privileged prefix of $w$, a contradiction.
\end{proof}

We now prove a result on powers and privileged words.

\begin{theorem}
Let $w$ be any word and $k$ an integer $\geq 1$.  If $w^k$
is privileged, then $w^j$ is privileged for all integers $j \ge 0$.
\end{theorem}

\begin{proof}
Suppose $k \geq 2$.  Then $w$ is a border of $w^k$, and hence
by Theorem~\ref{pelto} (b) we know $w$ is privileged.

It remains to show that if $w$ is privileged, then so is $w^j$
for all $j \geq 0$.  We prove this by induction on
$j$.  The result is clearly true for $j = 0$ or $j = 1$, so assume
$j \geq 2$ and $w^{j-1}$ is privileged.

Let $u$ be the longest proper privileged prefix of $w^j$.  
If $|u| \leq |w^{j-1}|$, then $u$ is also a privileged prefix of $w^{j-1}$.
Then Theorem~\ref{pelto} (a) and induction together imply that
$u$ is a suffix of $w^{j-1}$.   Then $u$ is also a suffix of $w^j$, and
by Lemma~\ref{fors} we know $w^j$ is privileged.

Otherwise $|u| > |w^{j-1}|$. Write $u=w^{j-1}y$
for some $y$, where $y$ is a proper prefix of $w$.  Since $j \geq 2$,
we see that $y$ is also a proper prefix of
$w^{j-1}$ and hence a proper prefix of $u$.  Thus $y$ is a border
of $u$, and hence, by Theorem~\ref{pelto} (b), $y$ is privileged.
Since $y$ is a privileged prefix
of $w$, by Theorem~\ref{pelto} (a), it is also a suffix of $w$.
Write $w=zy$ for some $z$. By induction we know that
$w^{j-1}$ is privileged.  Since $w^{j-1}$ is a prefix of $u$,
by Theorem~\ref{pelto} (a), it is also a suffix of $u$, so there exists
$x$ such that $u = x w^{j-1}$.  
Since $u=w^{j-1}y=xw^{j-1}$, we see that $|x| = |y|$ and $x$ is a 
proper prefix of $w$.  Thus in fact $x = y$.  So $u = y w^{j-1}$.
Then 
$$w^j=w \, w^{j-1}= (zy) w^{j-1}= z (y w^{j-1}) = zu,$$
and it follows that $u$ is a  suffix of $w^j$.   By Lemma~\ref{fors},
we conclude that $w^j$ is privileged.  This completes the induction.
\end{proof}

\section{The set of privileged words}

Let $\Sigma$ be a fixed alphabet and consider 
${\cal P}$, the set of privileged words over $\Sigma$.
We prove here that ${\cal P}$ is neither regular nor context-free.

\begin{proposition}
If $|\Sigma| \geq 2$, then ${\cal P}$ is not regular.
\end{proposition}

\begin{proof}
Let $0, 1$ be distinct letters in $\Sigma$.  Assume $\cal P$
is regular, and consider $L = {\cal P} \ \cap \ 0^+ 1 0^+$.
By Corollary~\ref{cor-leading} we have 
$L = \lbrace 0^n 1 0^n \ : \ n \geq 1 \rbrace$.
By the pumping lemma, $L$ is not regular, and hence neither is $\cal P$.
\end{proof}

\begin{proposition}
If $|\Sigma| \geq 2$, then ${\cal P}$ is not context-free.
\end{proposition}

\begin{proof}
Assume ${\cal P}$ is context-free, and consider 
the regular language $R=0^+10^+110^+$.
By a well-known closure property of the context-free languages,
$L := {\cal P} \ \cap \ R$ is context-free.
We will now use Ogden's lemma  \cite{Ogden} to show that
$L$ is not context-free, a contradiction.

We claim that 
$$L = \{00^a100^b1100^c \ : \  a=c \text{ and } a>b \}.$$

To see this, note that $L \subseteq R$. Thus it
suffices to show that a word $w$ of the form $0^{a+1} 1 0^{b+1} 11 0^{c+1}$
word is privileged if and only if $a=c$ and $a>b$.

($\Rightarrow$) Since $w$ begins and ends with $0$,
by Corollary~\ref{cor-leading}, we know that $a+1 = c+1$ and so $a = c$.
Suppose $b \ge a$. Then
$0^{a+1}10^{a+1}$ is a privileged prefix of $w$, yet it is not a
suffix of $w$.  By Theorem~\ref{pelto} (a), $w$ is not privileged.
Thus $a > b$.

($\Leftarrow$) Let $w=00^a100^b1100^a$ where $a>b$. Then the longest
proper privileged prefix of $w$ is $0^{a+1}$, which appears again as a
suffix of $w$. Thus $w$ is privileged.

Now let $n$ be as in Ogden's lemma, and let
$w=\underline{0^n} 10^{n-1}110^{n}$, where the first block of $n$ zeros
is marked as required by Ogden's lemma.  Then there exists some
decomposition $w=uxvyz$ where $xvy$ contains at most $n$ `marked'
characters, $xy$ contains at least 1 `marked' character, and $ux^ivy^iz
\in L$ for all $ i \ge 0$.

We see that if either $x$ or $y$ contain a 1, then $ux^0vy^0z$ will
have too few ones, and thus will not be in $L$.
Otherwise, we know $x$ lies entirely in the first block of zeros. If
$y$ does not lie in the last block of zeros, then if $i=0$, we will
have $a<c$, so $ux^0vy^0z \notin {\cal P} \cap R$. If $y$ does lie in
the last block of zeros, then $ux^0vy^0z=00^{n-j}100^{n-1}1100^{n-k}$
for some $j,k > 0$. Since $n-j \le n-1$, we see that $w \notin L$.

Hence no decomposition for $w$ exists with
$ux^0vy^0z \in L$, and thus ${\cal P} \cap R$ is not context-free.
Thus, the language of privileged words is not context-free.
\end{proof}

\section{A linear-time algorithm for determining if a word is 
privileged}

In this section we present an efficient algorithm for determining
if a given word is privileged.

Algorithm P:
\begin{algorithmic}
\Function{check-privileged}{$w$}
    	\If {$|w| \leq 1$}
    		\State \Return True
	\Else
		\State $T[0] \gets 0$
		\State $p \gets 1$
		\For{$i=1$ to $|w|-1$}
			\State $j \gets T[i-1]$
			\While{true}
				\If {$w[j] = w[i]$}
					\State $T[i] \gets j+1$
					\If {$T[i] = p$}
						\State $p \gets i+1$
					\EndIf
					\State exit while loop
				\ElsIf{$j = 0$}
					\State $T[i] \gets 0$
					\State exit while loop
				\EndIf
				\State $j \gets T[j-1]$
			\EndWhile
		\EndFor
		\If{$p = |w|$}
			\State \Return True
		\Else
			\State \Return False
		\EndIf
	\EndIf
\EndFunction
\end{algorithmic}

Our algorithm is a slightly modified version of the algorithm for
building a failure table in the well-known Knuth-Morris-Pratt
linear-time string-matching algorithm \cite{Knuth}.

\begin{theorem}
Algorithm P returns ``true'' if and only if $w$ is privileged. 
\end{theorem}

\begin{proof}
It is easy to see that if $|w|=0$ or $|w|=1$, then $w$ is privileged
and the algorithm returns ``true''. Otherwise, we consider the value for
$p$ at each iteration of the for-loop.

We now claim that at the end of each iteration of the for-loop, $p$
equals the length of the longest privileged prefix of the first $i+1$
characters of $w$.

To see the claim, observe that,
when entering the first loop we have $p=1$, and is the longest privileged
prefix of the first character of $w$. This establishes our base case.
Otherwise, we assume $p$ is the longest privileged prefix of the first
$i$ characters of $w$ at the beginning of the for loop, and prove our
claim for the end of this iteration. We note that $T[i]$ represents the
length of the longest subword $u$ which is both a prefix and suffix of
the first $i+1$ characters of $w$ (the word ``read so far''). If
$T[i]=p$, we know $u$ is privileged, and $p$ is increased to $i+1$.
Since $p$ is increased as soon as this equality is found, this is the
first time $u$ is repeated in $w$, and thus the word read so far is
privileged. This proves our claim.

After $w$ has been completely read by our algorithm, $p$ represents
the length of the longest privileged prefix of $w$. The algorithm
returns ``true'' if and only if $p=|w|$, in which case $w$ is privileged.
\end{proof}

Next, we have

\begin{theorem}
Algorithm P runs in $O(n)$ time, where $n = |w|$.
\end{theorem}

\begin{proof}
Starting with the KMP algorithm,
we have added one extra \emph{if} statement in the main loop,
allowing this algorithm to run in the same
$O(|w|)$ time bound as the original algorithm.  

More formally, we consider the number of times the inner while loop is
executed, as all else takes constant time. The first time the while
loop is executed, $i=1$ and $j=0$. Upon each iteration, we see that
either
\begin{enumerate}
	\item $i$ is incremented by 1, and $j$ is incremented by at most 1;
	\item $j$ decreases
\end{enumerate}
We see $i$ is incremented by exactly 1 when $w[j]=w[i]$ or $j=0$, due
to moving to the next iteration of the for loop.  When $j=0$, then $j$
will remain 0 beginning the next execution of the while loop. When
$w[i]=w[j]$, then $j$ will be set to $j+1$ in the next execution of the
while loop.

If neither of the above cases are fulfilled, we see $j$ is set to
$T[j-1]$, which is known by a property of the failure array to be
strictly less than $j$.

With these cases, we see that either $i$ increases or $i-j$ increases.
Since the algorithm terminates when $i=|w|-1$, $i$ will increase
exactly $n-2$ times, where $n = |w|$.
Also, since $j<i$ at each stage of the algorithm,
$i-j$ can increase at most $n-3$ times. Since these are the only
possible cases, the while loop will execute no more than $2n-5$ times.
Thus, Algorithm P takes $O(n)$ time to complete.
\end{proof}

\section{A lower bound on the number of privileged binary words}

Let $B(n)$ denote the number of privileged binary words of
length $n$.

We observe that if $x = 0^t 1 w 1 0^t$, and $w$ contains no
occurrences of $0^t$, then $x$ is privileged.  By choosing the
appropriate value of $t$, we get our lower bound.
First, though, we need a detour into
generalized Fibonacci sequences.

We need to count the number of words of length $n$ that contain
no occurrence of $0^t$.  As is well-known \cite[p.\ 269]{Knuth0}
and easily proved, this is
$G_n^{(t)}$, where
$$ G_n^{(t)} = \begin{cases}
	2^n, & \text{if $0 \leq n < t$}; \\
	G_{n-1}^{(t)} + G_{n-2}^{(t)} + \cdots + G_{n-t}^{(t)},
		& \text{if $n \geq t$}.
	\end{cases}
$$
We point out that
in the case where $t = 2$, this is $F_{n+2}$, the $(n+2)$'nd Fibonacci number,
where $F_0 = 0$, $F_1 = 1$, and $F_n = F_{n-1} + F_{n-2}$.  

It is well-known from the theory of linear recurrences that 
$$G_n^{(t)} = \Theta(\gamma_t^n),$$
where $1 < \gamma_t < 2$ is the root of the equation
$x^t - x^{t-1} - \cdots - x - 1 = 0$.  Since
$\gamma_t^t - \gamma_t^{t-1} - \cdots - \gamma_t - 1 = 0$,
multiplying by $\gamma_t - 1$ we get
$\gamma_t^{t+1} - 2 \gamma_t^t + 1 = 0$, so $\gamma_t = 2- \gamma_t^{-t}$.

The next step is to
find a good lower bound on $\gamma_t$.

\begin{lemma} 
Let $s \geq 2$ be an integer and let $\beta$ be a real number with
$0 \leq \beta \leq {6 \over s}$.  Then
$$ 2^s - \beta s 2^{s-1} \leq (2-\beta)^s .$$
\label{ineq-lem}
\end{lemma}

\begin{proof}
For $ s = 2$, the claim is $4-4\beta \leq (2-\beta)^2 = 4 - 4\beta + \beta^2$.
Otherwise, assume $s \geq 3$.  The result is clearly true for $\beta = 0$,
so assume $\beta > 0$.
By the binomial formula, we have
\begin{eqnarray}
(2-\beta)^s &=& \sum_{0 \leq i \leq s} 2^{s-i} (-\beta)^i {s \choose i} 
	\notag \\
&=& 2^s - \beta s 2^{s-1} + \sum_{2 \leq i \leq s} 2^{s-i} (-\beta)^i
	{s \choose i} \notag \\
&=& 2^s - \beta s 2^{s-1} +
	\sum_{1 \leq j \leq (s-1)/2} \left( 2^{s-2j} \beta^{2j} {s \choose {2j}}
		- 2^{s-2j-1} \beta^{2j+1} {s \choose {2j+1}} \right)
	\label{s1} \\
&+&	\begin{cases} \beta^s, & \text{if $s$ even}; \\
	0, & \text{otherwise.}  
	\end{cases} \notag
\end{eqnarray}

It therefore suffices to show that each term of the sum
\eqref{s1} is positive, or, equivalently, that
$$ 2^{s-2j} \beta^{2j} {s \choose {2j}} \geq 
2^{s-2j-1} \beta^{2j+1} {s \choose {2j+1}} .$$
for $1 \leq j \leq (s-1)/2$.

Now $\beta \leq {6 \over s}$ by hypothesis, so
$\beta \leq {6 \over {s-2}}$.  Hence $\beta s -2\beta \leq 6$.
Adding $2\beta - 2$ to both sides we get
$\beta s - 2 \leq 4 + 2\beta$,
and so ${{\beta s - 2} \over {2 + \beta}} \leq 2$.
If $i \geq 2 \geq {{\beta s - 2} \over {2 + \beta}}$ then
$(2+\beta)i \geq \beta s - 2$, so
$2(i+1) \geq \beta(s-i)$, and
$${2 \over \beta} 
\geq {{s-i} \over {i+1}} =
{{s \choose {i+1}} \over {s \choose i}} .$$
Thus $2 {s \choose i} \geq \beta {s \choose {i+1}}$.
Let $i = 2j$, and multiply both sides by $2^{s-2j} \beta^{2j}$ to get
$2^{s-2j} \beta^{2j} {s \choose {2j}} \geq 2^{s-2j-1} \beta^{2j+1} {s
\choose {2j+1}}$, which is what we needed.
\end{proof}

\begin{theorem}
Let $t \geq 2$ be an integer and define
$$\alpha_t = 2 - {1 \over {2^t - {t\over 2} - {t^2 \over {2^t}}}} .$$
Then $\alpha_t \leq 2 - \alpha_t^{-t}$.
\label{alph-thm}
\end{theorem}

\begin{proof}
It is easy to verify that
$$ {{3t^2} \over 4} \geq {{t^3} \over {2^t}} + {{t^4}\over {2^{2t}}} $$
for all real $t \geq 2$.  Hence
$$ 0 \leq {{3t^2} \over 4} - {{t^3} \over {2^t}} - {{t^4}\over {2^{2t}}},$$
and, adding $t2^{t-1}$ to both sides, we get
\begin{eqnarray*}
t 2^{t-1} &\leq& t 2^{t-1} +{{3t^2} \over 4} - {{t^3} \over {2^t}} - {{t^4}\over {2^{2t}}} \\
&=& \left( {t \over 2} + {{t^2} \over {2^t}} \right) \left( 2^t - {t \over 2}
- {{t^2} \over {2^t}} \right) .
\end{eqnarray*}
Setting 
$\beta_t = {1 \over {2^t - {t \over 2} - {{t^2} \over {2^t }}}} $,
we therefore have
$$ \beta_t t 2^{t-1} \leq {t \over 2} + {{t^2} \over {2^t}} ,$$
or
$$ -\beta_t t 2^{t-1} \geq -{t \over 2} - {{t^2} \over {2^t}}.$$
Add $2^t$ to both sides to get
$$ 2^t - \beta_t t 2^{t-1} \geq 2^t - {t \over 2} - {{t^2} \over {2^t}}.$$
Now it is easily verified that $\beta_t \leq 6/t$ for $t\geq 2$, so
we can apply Lemma~\ref{ineq-lem} with $s = t$ to get
$2^t - \beta t 2^{t-1} \leq (2-\beta)^t$.
It follows that
$$ (2-\beta)^t \geq 2^t - {t \over 2} - {{t^2} \over {2^t}},$$
and so
$$ \beta_t \geq (2-\beta_t)^{-t}.$$
It follows that
$$ 2 - \beta_t \leq 2 - (2-\beta)^{-t}.$$
Since $\alpha_t = 2- \beta_t$, we get
$$\alpha_t \leq 2 - \alpha_t^{-t},$$
as desired.
\end{proof}

We can now apply this to get a bound on $G_n^{(t)}$.

\begin{corollary}
Let $t \geq 2$ be an integer and $n \geq 0$.
Then $G_n^{(t)} \geq \alpha_t^n$, where
$\alpha_t = 2 - {1 \over {2^t - {t\over 2} - {{t^2} \over {2^t}}}} < 2$.
\end{corollary}

\begin{proof}
By induction on $n$.  Clearly $G_n^{(t)} = 2^n \geq \alpha_t^n$ for
$0 \leq n < t$ by definition.  Otherwise we have
\begin{eqnarray*}
G_n^{(t)} &=& G_{n-1}^{(t)} + \cdots + G_{n-t}^{(t)} \\
&\geq & \alpha_t^{n-1} + \cdots + \alpha_t^{n-t} \\
& = & {{ \alpha_t^n - \alpha_t^{n-t}} \over {\alpha_t - 1 }} .
\end{eqnarray*}
However, $\alpha_t \leq 2 - \alpha_t^{-t}$ by
Theorem~\ref{alph-thm}, so
$$\alpha_t - 1 \leq 1 - \alpha_t^{-t} .$$
Hence $(\alpha_t - 1) \alpha_t^n \leq (1-\alpha_t^{-t}) \alpha_t^n
= \alpha_t^n - \alpha_t^{n-t}$,
so from above we have
$$ G_n^{(t)} \geq {{\alpha_t^n - \alpha_t^{n-t}} \over {\alpha_t - 1}} \geq \alpha_t^n. $$
\end{proof}

Now we state and prove our lower bound on the number of binary
privileged words of length $n$.

\begin{theorem}
There are at least
$$ {{2^{n-5} } \over {n^2}}$$
privileged binary words of length $n$.
\label{bound}
\end{theorem}

\begin{proof}
Each word of the form $0^t 1 w 1 0^t$ is privileged, where
$|w| = n-2t-2$ and $w$ contains no factor $0^t$.  The number
of such $w$, as we have seen, is $G_{n-2t-2}^{(t)}$.
So it suffices to pick the right $t$ to
get a lower bound on $G_{n-2t-2}^{(t)}$.

It is easy to check, using the data in the next section,
that our bound holds for $n \leq 10$.  So assume $n \geq 11$.

Now
\begin{eqnarray*}
G_{n-2t-2}^{(t)} & \geq & \alpha_t^{n-2t-2} \\
&=& (2-\beta_t)^{n-2t-2} \\
&\geq& 2^{n-2t-2} - \beta_t (n-2t-2) 2^{n-2t-3} \\
& = & 2^{n-2t-2} (1- \beta_t (n/2 -t - 1) ) ,
\end{eqnarray*}
by Lemma~\ref{ineq-lem} with $s = n - 2t - 2$, provided
$\beta_t \leq 6/(n-2t-2)$. 

We now choose $t = \lfloor \log_2 n \rfloor + 1$, so that
\begin{equation}
2^{t-1} \leq n < 2^t .
\label{eq2t}
\end{equation}
It is now easy to verify that
$\beta_t \leq 6/(n-2t-2)$ for $n \geq 11$.   

On the other hand, it is easy to verify that
$${{3t} \over 4} \geq {{t^2} \over {2^{t+1}}} $$
for all real $t \geq 0$, so
$$ {{3t} \over 4} + 1 - {{t^2} \over {2^{t+1}}} > 0.$$
Adding $2^{t-1}$ to both sides, and using \eqref{eq2t}, we get
$$ {n \over 2} < 2^{t-1} < 2^{t-1} + {{3t} \over 4} + 1 - {{t^2} \over 2^{t+1}} ,$$
which implies
$$ {n\over 2} -t - 1 \leq {1 \over 2} \left( 2^t - {t \over 2} - {{t^2} \over {2^t}} \right)$$
and so
$\beta_t (n/2 - t - 1) \leq 1/2$.

It follows that 
$$ B(n) \geq G_{n-2t-2}^{(t)} \geq  2^{n-2t-2} (1- \beta_t (n/2 -t - 1) \geq
2^{n-2t-3} \geq {{2^{n-5}} \over {n^2}} .$$
\end{proof}

\begin{openproblem}
What is the true asymptotic behavior of $B(n)$ as $n \rightarrow \infty$?
\end{openproblem}

Define the function $f$ as follows:
$$f(n) = \begin{cases} 
	n, & \text{ if $n \geq 2$; } \\
	n f(\lfloor \log_2 n \rfloor), & \text{otherwise.} 
	\end{cases}
$$
It should be possible to improve Theorem~\ref{bound} to
$B(n) = \Omega(2^n c^{\log^{*} (n)}  /f(n))$, where $c$ is a constant
and, as usual,
$\log^{*} (n)$ is the number of times we need to apply $\log_2$ to
$n$ to get a number $\leq 1$.
We sketch the outline of an incomplete argument here:

We generalize our argument above to count the number of privileged
words of length $n$
having any privileged border of length $\lfloor \log_2 n \rfloor$.
We can use our previous argument provided the count for arbitrary
patterns is larger than the count for $0^t$.

More precisely,
if $x(p,n)$ is the number of strings of
length $n$ beginning with the pattern $p$, ending with $p$,
and having no other occurrence of $p$, then $x(p,n)$
satisfies a linear recurrence of order $t = |p|$.  By analyzing this
carefully, it should be possible to show
that, provided $n$ is in a certain range with respect to $|p|$,
we have $x(p,n) \geq x(0^t, n)$.  

Then we can imitate our analysis above,
setting $t = \lfloor \log_2 n \rfloor$, to get
\begin{eqnarray*}
B(n) &\geq &
\sum_{{p~\text{privileged}} \atop {|p| = \lfloor \log_2 n \rfloor}}  
x(p,n)  \\
& \geq  & c B(\lfloor \log_2 n \rfloor) \cdot {{2^n} \over {n^2}},
\end{eqnarray*}
for a constant $c$.    By iterating this relationship
$\log^{*} (n)$ times, we would get the claimed bound.

\section{Explicit enumeration of privileged words}

We finish with a table giving the number $B(n)$
of privileged binary words of length $n$ for
$0 \leq n \leq 38$.  It is sequence A231208
in Sloane's {\it On-line Encyclopedia of Integer Sequences}
\cite{Sloane}.

\begin{table}[H]
\begin{center}
\begin{tabular}{|c|c||c|c||c|c|}
\hline
$n$ & $B(n)$ & $n$ & $B(n)$ & $n$ & $B(n)$ \\
\hline
0 & 1 & 13 & 328 & 26 & 875408 \\
1 & 2 & 14 & 568 & 27 & 1649236 \\
2 & 2 & 15 & 1040 & 28 & 3112220 \\
3 & 4 & 16 & 1848 & 29 & 5888548 \\
4 & 4 & 17 & 3388 & 30 & 11160548 \\
5 & 8 & 18 & 38576 & 31 & 21198388 \\
6 & 8 & 19 & 71444 & 32 & 40329428 \\
7 & 16 & 20 & 133256 & 33 & 76865388 \\
8 & 20 & 21 & 248676 & 34 & 146720792 \\
9 & 40 & 22 & 466264 & 35 & 280498456 \\
10 & 60 & 23 & 875408 & 36 & 536986772 \\
11 & 108 & 24 & 1649236 & 37 & 1029413396  \\
12 & 176 & 25 & 3112220 & 38 & 1975848400 \\
\hline
\end{tabular}
\end{center}
\end{table}

\end{document}